\documentclass{llncs}
\usepackage{graphicx}
\usepackage{amsmath}
\usepackage{amssymb}
\usepackage{enumerate}
\usepackage{varioref}
\usepackage{charter,eulervm}

\newcommand{\es}{\varnothing}

\title{\sc {On the strong chromatic index and  
maximum induced matching 
of tree-cographs and permutation graphs}}
\author{
 Ton~Kloks\thanks{National 
Science Council of Taiwan Support Grant 
NSC~99--2218--E--007--016.}%
\inst{1}
\and 
 Chin-Ting~Ung\inst{1} 
\and 
 Yue-Li~Wang\inst{2}
}
\institute{
 Department of Computer Science\\
 National Tsing Hua University,
 No.~101, Sec.~2, Kuang Fu Rd., Hsinchu, Taiwan\\
 {\tt wonderboy0915@gmail.com}  
\and 
 Department of Information Management\\
 National Taiwan University of Science and Technology\\
 No.~43, Sec.~4, Keelung Rd., Taipei, 106, Taiwan\\
 {\tt ylwang@cs.ntust.edu.tw}
}

\begin{document}

\maketitle

\begin{abstract}
We show that there exist linear-time algorithms that compute the 
strong chromatic index and a maximum 
induced matching of tree-cographs when the decomposition tree 
is a part of the input. We also show that there exists an 
efficient algorithm for the strong chromatic index of 
permutation graphs.
\end{abstract}

\section{Introduction}
%%%%%%%%%%%%%%%%%%%%%%

\begin{definition}
Let $G=(V,E)$ be a graph. A {\em strong edge coloring\/} of $G$ is a 
proper edge coloring such that no edge is adjacent to two edges of the 
same color. 
\end{definition}

Equivalently, a strong edge coloring of $G$ is a vertex coloring 
of $L(G)^2$, the square of the linegraph of $G$. 
The strong chromatic index of $G$ is the minimal integer $k$ such 
that $G$ has a strong edge coloring with $k$ colors. We denote the 
strong chromatic index of $G$ by $s\chi^{\prime}(G)$. 

The class of tree-cographs was introduced by 
Tinhofer in~\cite{kn:tinhofer}. 

\begin{definition}
Tree-cographs are defined recursively by the following rules. 
\begin{enumerate}[\rm 1.]
\item Every tree is a tree-cograph.
\item If $G$ is a tree-cograph then also the 
complement $\Bar{G}$ of $G$ is a tree-cograph. 
\item For $k \geq 2$, if $G_1,\ldots,G_k$ are connected tree-cographs 
then also the disjoint union is a tree-cograph.
\end{enumerate}
\end{definition}

Let $G$ be a tree-cograph. A decomposition tree for $G$ consists 
of a rooted binary tree $T$ in which each 
internal node, including the root, 
is labeled as a join node $\otimes$ or a union node $\oplus$. 
The leaves of $T$ are labeled by trees or complements of trees. 
It is easy to see that a decomposition tree for a tree-cograph $G$ 
can be obtained in $O(n^3)$ time. 
 
\section{The strong chromatic index of tree-cographs}
%%%%%%%%%%%%%%%%%%%%%%%%%%%%%%%%%%%%%%%%%%%%%%%%%%%%%

The linegraph $L(G)$ of a graph $G$ is the intersection graph 
of the edges of $G$~\cite{kn:beineke}. 
It is well-known that, when $G$ is a tree then the linegraph $L(G)$ 
of $G$ is a claw-free blockgraph~\cite{kn:harary}. 
A graph is chordal if it has no induced 
cycles of length more than three~\cite{kn:dirac}. Notice that 
blockgraphs are chordal.  

\medskip

A vertex $x$ in a graph $G$ is simplicial if its neighborhood 
$N(x)$ induces a clique in $G$. Chordal graphs are characterized 
by the property of having a perfect elimination ordering, 
which is an ordering $[v_1,\ldots,v_n]$ of the vertices 
of $G$ such that $v_i$ is simplicial in the graph 
induced by $\{v_i,\ldots,v_n\}$. A perfect 
elimination ordering of a chordal graph can be computed 
in linear time~\cite{kn:rose}. This implies that 
chordal graphs have at most $n$ maximal cliques, and the 
clique number can be computed in linear time. 
 
\begin{theorem}[\cite{kn:cameron}]
If $G$ is a chordal graph then $L(G)^2$ is also chordal. 
\end{theorem}
\begin{proof}
Any chordal graph is the intersection graph of a collection 
of subtrees 
of a tree. Let $G$ be the intersection graph of a 
collection of subtrees of a tree. 
An intersection model for $L(G)^2$ is obtained by taking 
the union of every pair of intersecting subtrees. 
\qed\end{proof}

\begin{theorem}[\cite{kn:cameron3}]
\label{weakly chordal} 
Let $k \in \mathbb{N}$ and let $k \geq 4$. 
Let $G$ be a graph and assume that $G$ has no induced cycles 
of length at least $k$. Then $L(G)^2$ has no 
induced cycles of length at least $k$. 
\end{theorem}

\begin{lemma}
Tree-cographs have no induced cycles of length more than four. 
\end{lemma}
\begin{proof}
Let $G$ be a tree-cograph. 

First observe that trees are bipartite. It follows that complements of 
trees have no induced cycles of length more than four. 

We prove the claim by induction on the depth of a 
decomposition tree for $G$. If $G$ is the union of 
two tree-cographs $G_1$ and $G_2$ then the claim follows 
by induction since any induced cycle is contained in 
one of $G_1$ and $G_2$. 
Assume $G$ is the join of two tree-cographs $G_1$ and $G_2$. 
Assume that $G$ has an induced cycle $C$ of length at least five. 
We may assume that $C$ has at least one vertex in each of 
$G_1$ and $G_2$. If one of $G_1$ and $G_2$ has more than two 
vertices of $C$, then $C$ has a vertex of degree at least three, 
which is a contradiction. 
\qed\end{proof}

\begin{lemma}
\label{complement of tree}
Let $T$ be a tree. Then $L(\Bar{T})^2$ is a clique. 
\end{lemma}
\begin{proof}
Consider two non-edges $\{a,b\}$ and $\{p,q\}$ of $T$. 
If the non-edges share an 
endpoint then they are adjacent in $L(\Bar{T})^2$ since they 
are already adjacent in $L(\Bar{T})$. Otherwise, since $T$ is a 
tree, at least one pair of $\{a,p\}$, $\{a,q\}$, $\{b,p\}$ and 
$\{b,q\}$ is a non-edge in $T$, otherwise $T$ has a 4-cycle. 
By definition, $\{a,b\}$ and $\{p,q\}$ are adjacent in 
$L(\Bar{T})^2$. 
\qed\end{proof}

\medskip 

If $G$ is the union of two tree-cographs $G_1$ and $G_2$ 
then 
the maximal cardinality of a  clique in $L(G)^2$ 
is, simply, the maximum over the clique numbers 
of $L(G_1)^2$ and $L(G_2)^2$. 
The following lemma deals with the join of two tree-cographs. 

\begin{lemma}
\label{join}
Let $P$ and $Q$ be tree-cographs and let $G$ be the join of 
$P$ and $Q$. Let $X$ be the set of edges that have one endpoint 
in $P$ and one endpoint in $Q$. Then 
\begin{enumerate}[\rm (a)]
\item $X$ forms a clique in $L(G)^2$, 
\item every edge of $X$ is adjacent 
in $L(G)^2$ to every edge of $P$ and to every edge of $Q$, and 
\item every edge of $P$ is adjacent in $L(G)^2$ to every edge of $Q$.
\end{enumerate}
\end{lemma}
\begin{proof}
This is an immediate consequence of the definitions. 
\qed\end{proof}

\medskip 

For $k \geq 3$, a $k$-sun is a graph which consists of a clique 
with $k$ vertices and an independent set with $k$ 
vertices. There exist orderings $c_1,\ldots,c_k$ 
and $s_1,\ldots,s_k$ of the vertices in the clique 
and independent set such that each $s_i$ is adjacent to 
$c_i$ and to $c_{i+1}$ for $i=1,\ldots,k-1$ and such that 
$s_k$ is adjacent to $c_k$ and $c_1$.   
A graph is strongly chordal if it is chordal and has no 
$k$-sun, for $k \geq 3$~\cite{kn:farber}. 

\begin{figure}
\setlength{\unitlength}{1.8pt}
\begin{center}
\begin{picture}(170,40)
\thicklines
% sun 
\put(0,20){\circle*{2.0}}
\put(10,10){\circle*{2.0}}
\put(20,0){\circle*{2.0}}
\put(20,30){\circle*{2.0}}
\put(30,10){\circle*{2.0}}
\put(40,20){\circle*{2.0}}

\put(0,20){\line(1,-1){20}}
\put(0,20){\line(2,1){20}}
\put(10,10){\line(1,0){20}}
\put(20,0){\line(1,1){20}}
\put(40,20){\line(-2,1){20}}
\put(10,10){\line(1,2){10}}
\put(30,10){\line(-1,2){10}}

% gem 
\put(67,20){\circle*{2.0}}
\put(77,10){\circle*{2.0}}
\put(87,30){\circle*{2.0}}
\put(97,10){\circle*{2.0}}
\put(107,20){\circle*{2.0}}

\put(67,20){\line(1,-1){10}}
\put(67,20){\line(2,1){20}}
\put(77,10){\line(1,0){20}}
\put(77,10){\line(1,2){10}}
\put(107,20){\line(-2,1){20}}
\put(97,10){\line(-1,2){10}}
\put(97,10){\line(1,1){10}}

% claw
\put(130,5){\circle*{2.0}}
\put(150,5){\circle*{2.0}}
\put(150,25){\circle*{2.0}}
\put(170,5){\circle*{2.0}}
\put(150,25){\line(-1,-1){20}}
\put(150,25){\line(0,-1){20}}
\put(150,25){\line(1,-1){20}}
\end{picture}
\caption{A $3$-sun, a gem and a claw}
\label{3-sun}
\end{center}
\end{figure}
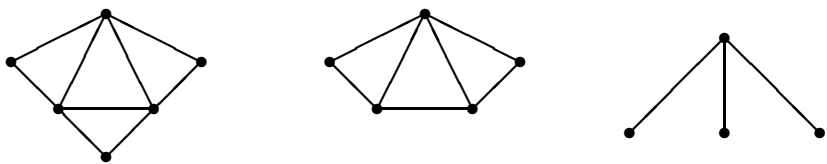

\begin{lemma}
\label{strongly chordal}
Let $T$ be a tree. Then $L(T)^2$ is strongly chordal. 
\end{lemma}
\begin{proof}
When $T$ is a tree then $L(T)$ is a blockgraph. 
Obviously, blockgraphs are strongly chordal. 
Lubiw proves in~\cite{kn:lubiw} that all powers of 
strongly chordal graphs are strongly chordal. 
\qed\end{proof}

We strengthen the result of Lemma~\ref{strongly chordal} as 
follows. 
Ptolemaic graphs are graphs that are both distance hereditary and 
chordal~\cite{kn:howorka}. Ptolemaic graphs are 
gem-free chordal graphs.  
The following theorem characterizes ptolemaic graphs. 

\begin{theorem}[\cite{kn:howorka}]
A connected graph is ptolemaic if and only if 
for all pairs of maximal cliques $C_1$ and $C_2$ with 
$C_1 \cap C_2 \neq \es$, the intersection $C_1 \cap C_2$ 
separates $C_1 \setminus C_2$ from $C_2 \setminus C_1$. 
\end{theorem}

\begin{lemma}
\label{ptolemaic}
Let $T$ be a tree. Then $L(T)^2$ is ptolemaic. 
\end{lemma}
\begin{proof}
Consider $L(T)$. Let $C$ be a block and let $P$ and $Q$ 
be two blocks that each intersect $C$ in one vertex. 
Since $L(T)$ is claw-free, the intersections of $P \cap C$ 
and $Q \cap C$ are distinct vertices.  
The intersection of the maximal cliques $P \cup C$ 
and $Q \cup C$, which is $C$, separates $P \setminus Q$ and 
$Q \setminus P$ in $L(T)^2$. Since all intersecting pairs 
of maximal cliques are of this form, this proves the lemma.  
\qed\end{proof}
     
\begin{corollary}
\label{char}
Let $G$ be a tree-cograph. Then $L(G)^2$ has a 
decomposition tree with internal nodes labeled as join nodes 
and union nodes and where the leaves are labeled as  
ptolemaic graphs. 
\end{corollary}

\medskip 

{F}rom Corollary~\ref{char} 
it follows that $L(G)^2$ is perfect~\cite{kn:chudnovsky}, 
that is, $L(G)^2$ has no odd holes or odd antiholes~\cite{kn:lovasz}. 
This implies that the chromatic number of $L(G)^2$ is equal 
to the clique number. Therefore, to compute the strong 
chromatic index of a tree-cograph $G$ it suffices to compute the 
clique number of $L(G)^2$.    

\begin{theorem}
Let $G$ be a tree-cograph and let $T$ be a decomposition tree 
for $G$. There exists a linear-time algorithm that computes 
the strong chromatic index of $G$. 
\end{theorem}
\begin{proof}
First assume that $G=(V,E)$ is a tree. Then the strong chromatic 
index of $G$ is 
\begin{equation}
\label{form1}
s\chi^{\prime}(G)=\max \;\{\; d(x)+d(y)-1 \;|\; (x,y) \in E\;\}
\end{equation}
where $d(x)$ is the degree of the vertex $x$. To see this 
notice that 
Formula~(\ref{form1}) gives the clique number of $L(G)^2$.  

\smallskip 

Assume that $G$ is the complement of a tree. 
By Lemma~\ref{complement of tree} the strong chromatic 
index is the number of nonedges in $G$, which is 
\[s\chi^{\prime}(G)=\binom{n}{2} - (n-1).\]

\smallskip 

Assume that $G$ is the union of two tree-cographs $G_1$ and 
$G_2$. Then, obviously, 
\[s\chi^{\prime}(G)= \max \;\{\; s\chi^{\prime}(G_1), 
\;s\chi^{\prime}(G_2)\;\}.\]  

\smallskip 

Finally, assume that $G$ is the join of two tree-cographs $G_1$ 
and $G_2$. Let $X$ be the set of edges of $G$ that have 
one endpoint in $G_1$ and the other in $G_2$. 
Then, by Lemma~\ref{join}, we have 
\[s\chi^{\prime}(G) = |X|+s\chi^{\prime}(G_1) + s\chi^{\prime}(G_2).\]

The decomposition tree for $G$ has $O(n)$ nodes. 
For the trees the strong chromatic index can be computed 
in linear time. In all other cases, the evaluation of 
$s\chi^{\prime}(G)$ takes constant time. It follows that 
this algorithm runs in $O(n)$ time, when a decomposition 
tree is a part of the input. 
\qed\end{proof}

\section{Induced matching in tree-cographs}
%%%%%%%%%%%%%%%%%%%%%%%%%%%%%%%%%%%%%%%%%%%

Consider a strong edge coloring of a graph $G$. Then each color class 
is an induced matching in $G$, which is an independent set 
in $L(G)^2$~\cite{kn:cameron}. 
In this section we show that the maximal value of an induced 
matching in $G$ can be computed in linear time. Again, we assume 
that a decomposition tree is a part of the input. 

\begin{theorem}
\label{induced matching}
Let $G$ be a tree-cograph and let $T$ be a decomposition 
tree for $G$. Then the maximal number of edges in an induced matching 
in $G$ can be computed in linear time. 
\end{theorem}
\begin{proof}
In this proof we denote the maximal cardinality of 
an induced matching in a graph $G$ 
by $i\nu(G)$. 

First assume that $G$ is a tree. Since the maximum induced matching 
problem can be formulated in monadic second-order logic, 
there exists a linear-time algorithm to compute the maximal 
cardinality of an induced matching in $G$. 

\smallskip 

Assume that $G$ is the complement of a tree. By 
lemma~\ref{complement of tree} $L(G)^2$ is a clique. 
Thus the maximal cardinality of an induced matching 
in $G$ is one if $G$ has a nonedge and otherwise it is zero.  

\smallskip 

Assume that $G$ is the union of two tree-cographs $G_1$ and 
$G_2$. Then 
\[i\nu(G) = i\nu(G_1)+i\nu(G_2).\]

\smallskip 

Assume that $G$ is the join of two tree-cographs 
$G_1$ and $G_2$. Then 
\[i\nu(G)= \max\;\{\;i\nu(G_1), \;i\nu(G_2), \;1\;\}.\]

This proves the theorem.
\qed\end{proof}

\section{Permutation graphs}
%%%%%%%%%%%%%%%%%%%%%%%%%%%%

A permutation diagram on $n$ points is obtained as follows. 
Consider two horizontal lines $L_1$ and $L_2$ in the Euclidean plane. 
For each line $L_i$ consider a linear ordering $\prec_i$ of 
$\{1,\ldots,n\}$ and put points $1,\ldots,n$ on $L_i$ in this order. 
For $k=1,\ldots,n$ connect the two points with the label $k$ by a 
straight line segment. 

\begin{definition}[\cite{kn:golumbic}]
A graph $G$ is a permutation graph if it is the intersection 
graph of the line segments in a permutation diagram. 
\end{definition}

Consider two horizontal lines $L_1$ and $L_2$ and on each line 
$L_i$ choose $n$ intervals. Connect the left - and right endpoint 
of the $k^{\mathrm{th}}$ interval on $L_1$ with the left - and 
right endpoint of the $k^{\mathrm{th}}$ interval on $L_2$. 
Thus we obtain a collection of $n$ trapezoids. We call this 
a trapezoid diagram. 

\begin{definition}
A graph is a trapezoid graph if it is the intersection 
graph of a collection of trapezoids in a trapezoid diagram. 
\end{definition}

\begin{lemma}
\label{permutation}
If $G$ is a permutation graph then $L(G)^2$ is a trapezoid graph. 
\end{lemma}
\begin{proof}
Consider a permutation diagram for $G$. Each edge of $G$ 
corresponds to  
two intersecting line segments in the diagram. The four endpoints 
of a pair of intersecting line segments define a trapezoid. 
Two vertices in $L(G)^2$ are adjacent exactly when the 
corresponding trapezoids intersect 
(see Proposition~1 in~\cite{kn:cameron2}). 
\qed\end{proof}

\begin{theorem}
There exists an $O(n^4)$ algorithm that computes a 
strong edge coloring in permutation graphs. 
\end{theorem}
\begin{proof}
Dagan, {\em et al.\/}, show that a trapezoid graph 
can be colored by a greedy coloring algorithm. 
It is easy to see that this algorithm can be adapted 
so that it finds a strong edge-coloring in permutation graphs. 
\qed\end{proof}

\begin{remark}
A somewhat faster coloring algorithm for trapezoid graphs 
appears in~\cite{kn:felsner}. Their algorithm runs in 
$O(n \log n)$ time where $n$ is the 
number of vertices in the trapezoid graph. 
An adaption of their algorithm yields a 
strong edge coloring for permutation graphs that runs 
in $O(m \log n)$ time, where $n$ and $m$ are the number of vertices 
and edges in the permutation graph. 
\end{remark}

\end{document}